\documentclass[onecolumn,journal]{IEEEtran}
\ifCLASSINFOpdf
\else
\fi

\usepackage{cite}
\usepackage{graphicx}
 \usepackage{amsmath}
\usepackage{diagbox}
 \usepackage{bm}
\usepackage{indentfirst}
\usepackage{caption}
\usepackage{algorithm}
\usepackage{algpseudocode}
\usepackage{booktabs}
\usepackage{subfigure}
\usepackage{multirow}
\usepackage{amsthm}
\usepackage{amsfonts}
\DeclareMathOperator*{\minimize}{minimize}

 \theoremstyle{remark}

\theoremstyle{}
\newtheorem{theorem}{Theorem}
\theoremstyle{}
\newtheorem{lemma}{Lemma}
\theoremstyle{}
\theoremstyle{}
\newtheorem{corollary}{Corollary}
\theoremstyle{}
\newtheorem{definition}{Definition}
\theoremstyle{remark}
\newtheorem{example}{Example}
\makeatletter
\newcommand{\tabcaption}{\def\@captype{table}\caption}
\makeatletter
\allowdisplaybreaks[4]

\begin{document}

\title{On the Placement and Delivery Schemes for Decentralized Coded Caching System}

\author{Qifa Yan, Xiaohu Tang, and~Qingchun Chen
\thanks{The authors are with The School of Information Science and Technology, Southwest Jiaotong University, 611756, Chengdu, China. E-mails: qifa@my.swjtu.edu.cn; xhutang@swjtu.edu.cn; qcchen@swjtu.edu.cn.
}
\thanks{The work was presented in part at WCSP 2016 \cite{greedy2016yan}.}
}

\maketitle

\begin{abstract}
Network based on distributed caching of content is a new architecture to alleviate the ongoing explosive demands for rate of multi-media traffic. In caching networks, coded caching is a recently proposed technique that achieves significant performance gains compared to uncoded caching schemes. In this paper, we derive a lower bound on the average rate with a memory constraint for a family of  caching allocation placement and a family of XOR cooperative delivery. The lower bound inspires us how placement and delivery affect the rate memory tradeoff. Based on the clues, we design a new placement and two new delivery algorithms. On one hand, the new placement scheme can allocate the cache more flexibly compared to grouping scheme. On the other hand, the new delivery can exploit more cooperative opportunities compared to the known schemes.  The simulations validate our idea.
\end{abstract}

\IEEEpeerreviewmaketitle

\section{Introduction}

With the dramatic increase of traffic, wireless networks are expected to support   higher and higher data rate applications such as video streaming.  However, the rate of the network is fundamentally restricted by its limited amount of spectrum, which leads to congestion during peak times. Fortunately,
the emergence of such applications involves repeated wireless transmissions of content that are requested multiple times by different users at different times.
Therefore, one promising approach to mitigate the congestion is ``removing" some of the traffic by duplicating some contents into memories distributed across the network when the network resources are abundant. Normally, this duplication is called content placement or caching and can be performed during off peak times.

As a result, the congestion during peak times can be reduced by utilizing caching in two separated phases: a content placement phase and a content delivery phase. When the network resources are under-utilized, the placement phase can be performed to disseminate some contents into memories across the network, even without knowledge of user request. Within the delivery phase in congested network scenarios, the user requests can be satisfied with the help of the contents in the memories.
In the past decades, many cache assisted delivery schemes have been investigated (e.g. \cite{Meyerson2001,Baev2008,Borst2010}). Among them, typically
coding across the cached content or/and the delivered signals was not considered, whereas
usually the utilization of popularity of files was focused on. Through efficient use of caching, their gains are derived from the local contents of individual user.

Recently,  Maddah-Ali and Niesen \textit{et al.} proposed coded caching schemes that create multicast opportunities by globally utilizing the caching of all users (e.g.,\cite{maddah2013decentralized,maddah2013fundamental,niesen2013coded}). They showed that significant gains can be achieved by coding in the delivery phase compared to the non coded schemes. The initial work was  investigated for a centralized setting \cite{maddah2013fundamental}, where the central server is able to coordinate all the contents at the users carefully to create multicast opportunities for the delivery phase. Later on, the work was extended to a decentralized setting \cite{maddah2013decentralized}, where the cached contents should be independent across different users. In contrast to the former,  the  advantage of the decentralized coded caching scheme is that, it can be realized in more circumstances, such  that the server does not know the identity  of the active users  or the users are in distinct networks in the placement phase, while its performance is close to its centralized counterpart \cite{Yan2016Gap}.
Moreover, various settings were explored, and the examples include online cache update \cite{maddah2013online,Yan2017Online}, hierarchical network \cite{Hierarchcal2014}, and many aspects  like subpacketization \cite{finite2016,Yan2015PDA,Yan2017bipartite}, delay \cite{NiesenDelay},  have been investigated.
However, there are some noteworthy shortcomings in the seminal decentralized coded caching scheme and the subsequent ones:

\begin{enumerate}
  \item The decentralized coded caching algorithm in \cite{maddah2013decentralized} evenly allocates each user's cache to the files.   Then it visits each user subset in order, and pads the bit vectors intended for different users to the length of the longest vector when visiting each user subset.
     In fact,  this algorithm is particularly suitable for uniform popularity because the ``padded zeros" do not contain any information, which leads to a loss in terms of rate-memory tradeoff. Even for uniform popularity,  the loss is significant in practical implementations as well, since the law of large numbers plays out.

  \item To cope with nonuniform popularity, Niesen and Maddah{-}Ali further proposed grouping algorithm in \cite{niesen2013coded}, where the files are split into several groups such that those having similar popularity are in the same group and allocated the same cache size. Accordingly, decentralized coded caching scheme is then implemented in those users whose requested files belonging to the same group. Since users are split into user groups based on the popularity of their requested files, nonuniform popularity is indeed transformed into near-uniform popularity in each group through grouping. However, users in distinct groups are served in turn, thus it excludes the cooperative opportunities between users in different groups, which  definitely deteriorates the rate-memory tradeoff performance.
\end{enumerate}

In this paper, addressing the decentralized coded caching scheme, we first identify a family of placement algorithms as \textit{cache allocation placement} and a family of delivery algorithms as \emph{XOR cooperative delivery}. These two generic families include the placement or delivery procedures of the decentralized coded caching \cite{maddah2013decentralized} and grouping \cite{niesen2013coded} algorithm as special cases. Under this framework, we derive a lower bound on the average delivery rate with  memory constraints. The lower bound inspires us on how the placement and the delivery affect the rate-memory tradeoff performance. Based on the clues obtained from the lower bound, we design a new placement algorithm and two new delivery algorithms.  In the new placement algorithm, we optimize the cache allocation parameters according to different files' popularity.  The new delivery algorithms, which we refer to as \emph{set centered greedy coded delivery} and \emph{bit centered greedy coded delivery}, overcome the two losses mentioned above through achieveing more cooperative opportunities. Our simulation results verify that in most cases, our new algorithm can achieve better performance, especially when $F\ll 2^K$.

The rest of this paper is  organized as follows: In Section \ref{sec_prob}, we introduce   the system model  and background information. In Section \ref{sec_lowerbound}, we derive a lower bound on the average  delivery rate for a family of placement and delivery algorithms. We  propose our new placement and delivery strategy in Section \ref{allocation} and \ref{delivery} respectively. Section \ref{sec_simu} presents our simulation results and Section \ref{sec_con} concludes this paper.

\textbf{Notations:} $\mathbb{E}[\cdot]$ is used to take expectation;  $\oplus$ is used to represent bitwise XOR operation; $\mathbb{P}(A)$ means the probability of event $A$; $|A|$ means the cardinality of set $A$; A column vector is denoted by a bold lower case character; A set (subset) of users or files are denoted in calligraphy font.

\section{Problem Setting and Background}\label{sec_prob}
\subsection{Problem Setting}
We consider a  system consisting of a server connected to $K$ users through an error-free shared link, as illustrated in Fig. \ref{cache}. The server has $N$ files, which are denoted by $\mathcal{N}=\{f_1,f_2,\cdots,f_N\}$, each file is  of  size $F$ bits. Since in a video application such as Netfix, the number of
“files” (each may correspond to a short video segment) is likely to be larger than the number of users, we follow \cite{Hierarchcal2014} to assume $N\geq K$. Suppose that each user has access to a cache of size $MF$ bits, where $M\in[0,N]$. The probability of file $f_i$ being requested by a user is $ p_i$, which is assumed to be  independent of the users in this paper. Let $\bm p=(p_1,p_2,\cdots,p_N)^T$,  we call this vector popularity.
\begin{figure}[htbp]
\centering\includegraphics[width=0.45\textwidth]{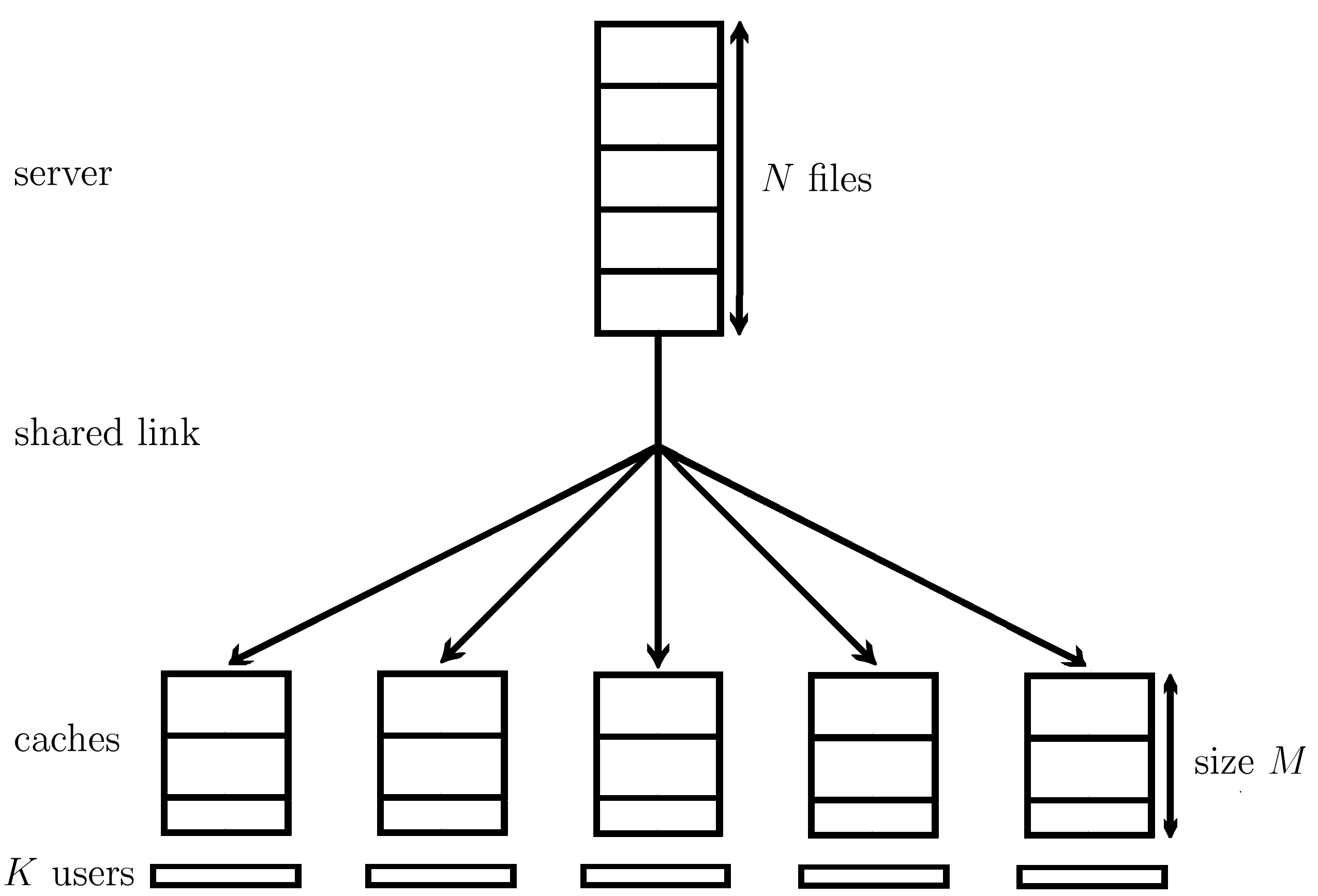}\caption{A server with $N$ files (each of $F$ bits) is connected to $K$ users equipped with an isolated cache of size $MF$ bits through a shared link. In the figure, $N=5, K=5$ and $M=\frac{5}{2}$.}\label{cache}
\end{figure}

The system operates in two phases: a placement phase and a delivery phase. The placement phase can be accomplished when the network load is low. During this phase, via the shared link the server is able to fill each user's cache as an arbitrary function of $N$ files and their popularity $\bm p$  only subject to its cache size $MF$ bits.  Note that the caching function is not allowed to depend on any user request. In other words, all the information about user requests available in this phase is the files' popularity $\bm p$ \cite{maddah2013fundamental}. Besides,  the decentralized setting requires that the cache contents of different users  are independent  \cite{maddah2013decentralized}.  During the delivery phase, it is assumed that each user requests one file from the server according to the popularity $\bm p$ independently. Precisely, the request of users is denoted by $\bm d=(d_1,d_2,\cdots,d_K)^T$, namely, user $k$ requests the $d_k$th file $f_{d_k}$, where $d_1,d_2,\cdots,d_K$ are  identically independently distributed (i.i.d) as
\begin{align}
\mathbb{P}(d_k=i)=p_i,~\forall ~k\in\mathcal{K}, ~1\le i\le N,\notag
\end{align}
and $\mathcal{K}=\{1,2,\cdots,K\}$ denotes the set of users.
Once the server received the users' request $\bm{d}$, it broadcasts $R_{\bm{d}}F$ bits to users. Each user is able to recover its requested file from the message received in the delivery phase with help of its caching.

We refer to $R_{\bm{d}}$ as the rate of the system under the condition of users' request $\bm{d}$. Since $\bm{d}$ is a random vector, the average rate $\mathbb{E}\left[R_{\bm d}\right]$ can be used to characterize the system performance, which is the main concern throughout this paper. Especially, we are interested in the achievable tradeoff between $\mathbb{E}\left[R_{\bm d}\right]$ and $M$. In order to simplify the notation, we drop the subscript $\bm d$, and then use $R$ to indicate  the rate of a realization of the system. Thus, the average rate of the system is denoted by $\mathbb{E}\mathbf[R]$.

\subsection{Related Work}\label{sec_prework}

In \cite{maddah2013decentralized}, Maddah-Ali and Niesen  initialized the ``coded caching" for decentralized setting. When $N\geq K$, they showed that significant gain can be achieved by using caches to create multicast opportunities. Concretely, they
  proved  that the following worst case (over all possible demands) rate-memory tradeoff curve
  \begin{align}
R(M,N,K)=\frac{1-{M}/{N}}{M/N}\cdot{\left(1-\left(1-\frac{M}{N}\right)^K\right)}\label{Rd}
\end{align}
can be arbitrarily approached as $F\rightarrow\infty$.

Note that rate \eqref{Rd} is approached by allocating the cache evenly to $N$ files. For popularity close to uniform popularity, it is reasonable. While for nonuniform demands, a strategy was introduced in \cite{niesen2013coded} that partitions the $N$ files into $L$ groups $\mathcal{N}_1,\mathcal{N}_2,\cdots,\mathcal{N}_L$, and each cache is partitioned into $L$ parts accordingly, one part for each group. By implementing the decentralized algorithm \cite{maddah2013decentralized} for each group, the algorithm approaches the following average rate
\begin{align}
R_G(M)=\sum_{l=1}^L\mathbb{E}\left[R(M_l,N_l,\mathsf{K}_l)\right]\label{Rg}
\end{align}
 when $F\rightarrow\infty$, where $R(\cdot,\cdot,\cdot)$  is given by \eqref{Rd}, $M_l$ is the (normalized) cache size allocated to group $\mathcal{N}_l$, $N_l$ is the number of files in group $\mathcal{N}_l$, and $\mathsf{K}_l$ is the number of users  who request a file in group $\mathcal{N}_l$. We highlight the fact that $\mathsf{K}_l$ is a random variable by writing it in sanf-serif font in \eqref{Rg}, where  the expectation is taken with respect to $\mathsf{K}_l$.

The  memory allocation parameters $M_l$ can be optimized to minimize the expectation in \eqref{Rg}, but finding  the optimal number of groups and allocating the files to these groups is not an easy task. In \cite{niesen2013coded},  a potentially suboptimal grouping strategy   was proposed,
 which is denoted by ``grouping algorithm" in  this paper.

For simplicity, in the following, we refer to the placement procedures of decentralized coded caching and grouping algorithms as \emph{original placement} (OP) and \emph{grouping placement} (GP), and the delivery procedures  as \emph{original delivery} (OD) and \emph{grouping delivery} (GD) respectively.

\section{A Lower Bound on Cooperative Delivery Rate}\label{sec_lowerbound}

In this section, we derive a lower bound on $\mathbb{E}[R]$ that can be applied to not only the known algorithms, i.e., the original one \cite{maddah2013decentralized} and the grouping one \cite{niesen2013coded}, but also general cases, i.e., \emph{caching allocation placement} and \emph{XOR cooperative delivery}.

\textbf{
Placement:} Assume in the placement phase, the server fills the users' cache by the following strategy: for each $k\in \mathcal{K}$ and each file $f_i\in\mathcal{N}$, user $k$ independently caches a subset of $q_iF$ bits of file $f_i$, chosen uniformly at random. The parameters $q_i$ should satisfy:
 \begin{align}
 &0\leq q_i\leq 1,~\forall ~1\leq i\leq N,\label{qleq1}\\
 &\sum_{i=1}^Nq_i=M.\label{qsum}
 \end{align}
    We call this \textit{cache allocation placement}.

 Note that, the above placement   will reduce to OP when $q_i=\frac{M}{N}$ for all $i$ and GP  when $q_i=\frac{M_l}{N_l}$,  $\forall~i\in \mathcal{N}_l$. In general, we can specify a particular placement strategy by correspondingly setting the allocation parameters $q_i$.

\textbf{ Delivery:} We assume that for every link use, the server transmits signals for some users $\mathcal{S}\subset \mathcal{K}$. The transmitted signals are of form $\oplus_{k\in \mathcal{S}}v_k$, where $\oplus$ is Exclusive OR (XOR) operation, $v_k$ represents a bit to be delivered to user $k$.  User $k$ caches all  $v_j, \forall~j\in \mathcal{S},j\ne k$, thus user $k$ subtracts other $v_j$ from the received signal $\oplus_{k\in \mathcal{S}}v_k$ to immediately decode $v_k$. We call this type of delivery strategy \textit{XOR cooperative delivery}.

 Obviously, the above delivery includes OD and GD as special cases.
For ease of illustration, we give several definitions.
 \begin{definition} \label{def2}
 For a bit $b$, we say that $b$ covers a user $k$  if  $b$ is in the cache of user $k$ and $b$ covers a set $\mathcal{T}$
 if $b$ covers every user in $\mathcal{T}$. In particular, the largest subset of $\mathcal{K}$ that $b$ covers
 is said to be  the cover set of $b$.
 \end{definition}

If $b$ covers $\mathcal{T}$, then when $b$ is transmitted after XOR operation with some bits intended for users in $\mathcal{T}$, those users can subtract $b$ from the received signal since $b$ is in their caches. Therefore, ``$b$ covers $\mathcal{T}$" can be interpreted as ``transmitting $b$ does not interfere any user in $\mathcal{T}$" from the point view of delivering bits.

 \begin{definition}
 In an XOR cooperative delivery scheme, if a bit $b$ is transmitted after XOR operation with other $k-1$ bits to be delivered to other $k-1$ users respectively, we say that $b$ is transmitted in $k$-multicast, and the (nonnormalized) rate of $b$ is $\frac{1}{k}$.
\end{definition}
Note that since the time slot is spent on transmitting $k$ bits,  rate $\frac{1}{k}$ essentially implies the efficiency of transmitting  $b$. On the other hand, if we aim to count time slots the system consumes, we can simply sum the rate of every bit the server transmits. Obviously, a larger $k$ indicates a lower rate and a higher efficiency.

 The lower bound result is summarized in the following theorem:
 \begin{theorem}\label{thm1}
 Assume that the server adopts a cache allocation placement  and an XOR cooperative delivery procedure described above, then the expected rate that the system can achieve in the delivery phase has the following lower bound:
 \begin{align}
 \mathbb{E}[R]\geq \sum_{i=1}^N\frac{1-q_i}{q_i}\left(1-(1-q_i)^K\right)p_i.\label{ER}
 \end{align}
 \end{theorem}
 \begin{proof}
   Assume that  the user request vector is $\bm{d}=(d_1,d_2,\cdots,d_K)^T$, i.e., user $k$ requests file $f_{d_k}$, where $d_k$ is assumed to be  independently identically distributed according to
 \begin{align}
 \mathbb{P}(d_k=i)=p_i,~ \forall ~1\le i\le N.\notag
 \end{align}

 Under the condition  $\bm{d}$, we know that user $k$ has access to $q_{d_k}F$ bits of its requested file in its cache, then there are $(1-q_{d_k})F$ bits needed to be transmitted for user $k$. Randomly label these bits as $1,2,\cdots,(1-q_{d_k})F$, and denote $r_{k,j}$ the rate of $j$th bit of user $k$, then the whole time slot the system consumes  is
 \begin{align}
 RF=\sum_{k=1}^K\sum_{j=1}^{(1-q_{d_k})F}r_{k,j}.\notag
 \end{align}
 Taking expectation under condition $\bm {d}$, due to the symmetry between different bits in the same file, we have
 \begin{align}
 \mathbb{E}[R|\bm{d}]F&=\sum_{k=1}^K\sum_{j=1}^{(1-q_{d_k})F}\mathbb{E}[r_{k,j}|\bm{d}]\notag\\
 &=\sum_{k=1}^K(1-q_{d_k})F\mathbb{E}[r_{k,1}|\bm {d}].\label{ER1}
 \end{align}
Thus,
 \begin{align}
 \mathbb{E}[R|\bm {d}]&=\sum_{k=1}^K(1-q_{d_k})\mathbb{E}[r_{k,1}|\bm {d}],\notag
 \end{align}
whose expectation with respect to $\bm{d}$ yields
 \begin{align}
 \mathbb{E}[R]&=\mathbb{E}\left[\mathbb{E}[R|\bm{d}]\right]\notag\\
 &=\sum_{k=1}^K\mathbb{E}\left[(1-q_{d_k})\mathbb{E}[r_{k,1}|\bm{d}]\right]\notag\\
 &=K\mathbb{E}\left[(1-q_{d_1})\mathbb{E}[r_{1,1}|\bm{d}]\right]\label{ER2}\\
 &=K\mathbb{E}\left[(1-q_{d_1})\mathbb{E}[r_{1,1}|d_1,d_2,\cdots,d_K]\right]\label{ER4}\\
 &=K\mathbb{E}\left[(1-q_{d_1})\mathbb{E}[r_{1,1}|d_1]\right]\label{ER3}
 \end{align}
 Equation \eqref{ER2} is due to the symmetry between different users and equation \eqref{ER3} is true since we take the expectation with respect $d_2,d_3,\cdots,d_K$ in the outer operation ``$\mathbb{E}$" of \eqref{ER4}, so the  outer expectation in \eqref{ER3} is only taken with respect to $d_1$.

 Assume that the cover set's  cardinality  of the user $1$'s first bit  is $\Lambda$. The value of $\Lambda$ lies in $\{0,1,2,\cdots,K-1\}$, since user 1 must not cache it. Note that, conditioned on $d_1=i$, $\Lambda$ obeys the binomial distribution, i.e,
 \begin{align}
 p(\Lambda=k|d_1=i)={K-1\choose k}q_i^k(1-q_i)^{K-1-k},~\forall ~k\in\{0,1,\cdots,K-1\}.\notag
 \end{align}
 Thus
 \begin{align}
 &\quad\mathbb{E}[r_{1,1}|d_1=i]\notag
 \\&=\sum_{k=0}^{K-1}\mathbb{E}[r_{11}|d_1=i,\Lambda=k]p(\Lambda=k|d_1=i)\label{ER5}\\
 &\geq\sum_{k=0}^{K-1}\frac{1}{k+1}{K-1\choose k}q_i^k(1-q_i)^{K-1-k}\label{ER6}\\
 &=\frac{1}{Kq_i}\sum_{k=1}^{K}{K\choose k}q_i^{k}(1-q_i)^{K-k}\notag\\
 &=\frac{1}{Kq_i}\left(1-(1-q_i)^K\right),\label{ER7}
 \end{align}
 where in \eqref{ER5}, we use the law of total expectation; \eqref{ER6} is due to the fact that under the condition $\Lambda=k$, namely the bit covers exactly $k$ users of $\mathcal{K}\backslash\{1\}$, the rate of the bit can not be lower than $\frac{1}{k+1}$ for  any XOR cooperative delivery algorithm; in \eqref{ER7}, we uses the binomial theorem.

  According to the equation \eqref{ER3} and \eqref{ER7}, we obtain the result:
  \begin{align}
  \mathbb{E}[R]&\geq K\cdot\sum_{i=1}^N\frac{1-q_i}{Kq_i}\left(1-(1-q_i)^K\right)p_i\notag\\
  &=\sum_{i=1}^N\frac{1-q_i}{q_i}\left(1-(1-q_i)^K\right)p_i.\notag
  \end{align}
 \end{proof}


  From  Theorem \ref{thm1} and its proof, two principles should be noted:
  \begin{enumerate}
    \item [P1.] The lower bound \eqref{ER} relies on the cache allocation parameters $q_i$. Thus, in order to let $\mathbb{E}[R]$ be as small as possible, in the placement phase, adjusting $q_i$ is useful. Furthermore, the lower bound can be seen as an approximation for $\mathbb{E}[R]$. In Section \ref{allocation}, we optimize the lower bound, and obtain a potentially suboptimal cache allocation strategy.
    \item [P2.] Note that the only inequality in the above induction is \eqref{ER6}. Then, the equality in \eqref{ER} holds only when in each realization, each bit transmitted by the server achieves the full multicast, i.e., for any bit $b$, it cooperates with every user in its cover set. Therefore, in order to make the lower bound as tight as possible, in the delivery phase, the server should try to make all bits transmitted as full multicast as possible. In Section \ref{delivery}, we introduce delivery algorithms  based on this principle.
  \end{enumerate}

The following corollary indicates that the lower bound is tight in some cases.
\begin{corollary}
When the cache is evenly allocated to each file, i.e., $q_i=\frac{M}{N}$, $i=1,2,\cdots,N$, and $F\rightarrow\infty$,  the lower bound in \eqref{ER} can be arbitrarily  approached, i.e., the optimal average rate is given by
\begin{align}
\mathbb{E}[R^*]=\frac{1-M/N}{M/N}\left(1-\left(1-\frac{M}{N}\right)^K\right).\label{ERa1}
\end{align}
\end{corollary}
\begin{proof}
Let $q_i=\frac{M}{N}$, $i=1,2,\cdots,N$, then by \eqref{ER}, we have the righthand side (r.h.s.) of \eqref{ERa1}. On the other hand, as stated in Section II,   the worst rate $R_{\mbox{worst}}$  given by \eqref{Rd} can be arbitrarily approached when $F\rightarrow\infty$. Thus,  we finish the proof by means of the fact $R_{\mbox{worst}}\geq \mathbb{E}[R]$.
\end{proof}

\section{A New Placement Strategy}\label{allocation}
Given the lower bound \eqref{ER}, which is expressed  with parameters $q_i$, we can  obtain the optimal cache allocation $q_i$. The following two lemmas lay the basis for the optimization problem.

\begin{lemma}\label{lem1} The function
\begin{eqnarray}\label{Eqn_h}
h(x)=\frac{x^2}{1-(1-x)^K(1+Kx)}
\end{eqnarray}
is monotonously increasing within the interval $(0,1]$, and $h(x)<1,\forall~x\in(0,1)$ with
\begin{eqnarray*}
\lim_{x\rightarrow 0^+}h(x)=\frac{2}{K(K+1)}.
\end{eqnarray*}
\end{lemma}

\begin{lemma}\label{lem2}
The function $$f(x)=\left\{\begin{array}{ll}
                                                        \frac{1-x}{x}\left(1-(1-x)^K\right),& x\in(0,1]  \\
                                                         K,& x=0
                                                      \end{array}
\right.$$ is a continuous, monotonously, decreasing and convex function within $[0,1]$.
\end{lemma}
 The proofs of the Lemma \ref{lem1} and \ref{lem2} are given in Appendix \ref{appendix1}. Now we can find the optimal $q_i$ in the sense that it minimizes the lower bound of \eqref{ER}. The optimal $q_i$s are just the solution to the following optimization problem:
 \begin{align}
 \minimize_{q_1,q_2,\cdots,q_N}\quad\sum_{i=1}^N\frac{1-q_i}{q_i}&\left(1-(1-q_i)^K\right)p_i\label{opt}\\
 \mbox{subject to}\qquad q_i-1&\leq 0,\label{ieq1}\\
 -q_i&\leq 0,~i=1,2,\cdots,N,\label{ieq2}\\
 \sum_{i=1}^Nq_i&=M.\label{eqM}
 \end{align}
 This is a standard convex optimization problem, whose associated Lagrange function is given by
 \begin{align}
 &L(q_1,\cdots,q_N, \lambda_1,\cdots,\lambda_N,\delta_1,\cdots,\delta_N,\nu)\notag\\
 =&\sum_{i=1}^N\frac{1-q_i}{q_i}\left(1-(1-q_i)^K\right)p_i+\sum_{i=1}^N\lambda_i(q_i-1)-\sum_{i=1}^N\delta_iq_i+\nu\left(\sum_{i=1}^Nq_i-M\right),\notag
 \end{align}
where $\lambda_i\geq0, \delta_i\geq0$ are the dual variables corresponding to inequality constraints \eqref{ieq1} and \eqref{ieq2} respectively, $\nu$ is the dual variable corresponding to the equality constraint \eqref{eqM}. For convex optimization, a point is optimal if and only if it satisfies the Karush-Kuhn-Tucker (KKT) \cite{boyd2009convex} conditions. By analyzing these conditions, we  have the following observations:
 \begin{enumerate}
   \item If $0<q_i<1$,
   \begin{align}
   p_i=\frac{q_i^2\nu}{1-(1-q_i)^K(Kq_i+1)}=h(q_i)\nu.\notag
   \end{align}
   By Lemma \ref{lem1}, $\frac{2\nu}{K(K+2)}<p_i<\nu$.
   \item If $q_i=1$,
   \begin{align}
   p_i=\nu+\lambda_i\geq\nu.\notag
   \end{align}
   \item If $q_i=0$,
   \begin{align}
   p_i=\frac{2}{K(K+1)}\nu-\frac{2}{K(K+1)}\delta_i\leq\frac{2}{K(K+1)}\nu.\notag
   \end{align}
 \end{enumerate}

 Theorem \ref{thm2} summarizes the results of the above analysis:
 \begin{theorem}\label{thm2} The optimal point of the optimization problem \eqref{opt} is given by
 \begin{align}
 q_i=\left\{\begin{array}{ll}
              1, & p_i\geq\nu  \\
              g\left(\frac{p_i}{\nu}\right), &\frac{2\nu}{K(K+1)}< p_i<\nu \\
              0 ,&p_i\leq \frac{2\nu}{K(K+1)}
            \end{array}
 \right.,\label{solution}
 \end{align}
 where being the inverse function of $h(\cdot)$ given in \eqref{Eqn_h}, $g(x)$ is uniquely well defined when $x\in\left[\frac{2}{K(K+2)},1\right]$ by Lemma \ref{lem1},
 and $\nu$ satisfies equation \eqref{eqM}.
 \end{theorem}

 If $\nu\leq0$, then $q_i=1$ for all $i$, thus $\sum_{i=1}^Nq_i=N$, which can happen only if $M=N$. Therefore, for $M<N$, we have $\nu>0$.   Since $q_i$ is a non-increasing function of $\nu$, we can first compute $\nu$  through bisection method by \eqref{eqM}, and then determine $q_i$ through \eqref{solution}.
 Moreover, $\nu$ in \eqref{solution} determines two thresholds: $\nu$ itself is a higher threshold, if the popularity of a file is larger than $\nu$, the whole file should be cached in every user; $\frac{2\nu}{K(K+1)}$ is a  lower threshold, if the file popularity is smaller than $\frac{2\nu}{K(K+1)}$, there is no need allocating any cache for that file; if the file's popularity lies between $\left[\frac{2\nu}{K(K+1)},\nu\right]$, the size allocated to the file should be calculated through a nonlinear transformation of the popularity.

 Furthermore, it is worthy noting that the cache allocation strategy in \eqref{solution} relies on the parameter $K$. In the case that the server does not have the knowledge of users' number in the placement phase, the strategy in \eqref{solution} can be approximately achieved for most $K$ as follows: observing that $h(x)\rightarrow x^2$ for any $x\in [0,1]$ when $K\rightarrow\infty$ (where $h(0)$ is defined as the left limit $\frac{2}{K(K+1)}$),  we  replace $h(x)$ with $u(x)=x^2$ and replace $g(x)$ with $v(x)=x^{\frac{1}{2}}$. Thus, we allocate the cache according to
 \begin{align}
 q_i=\left\{\begin{array}{cc}
               1,& p_i\geq \nu  \\
               \sqrt{\frac{p_i}{\nu}},& p_i<\nu
            \end{array}
 \right.,\label{solution2}
 \end{align}
 where $\nu$ is the solution of
 \begin{align}
 \sum_{i=1}^N\min\left\{\sqrt{\frac{p_i}{\nu}},1\right\}=M.\notag
 \end{align}

\section{New Delivery Strategies }\label{delivery}
In this section, we introduce our new delivery strategies.
We begin with a definition and an illustrative example.

\begin{definition}
Given a request $\bm{d}$, a bit $b$ is said to be intended for user $k_b$  if  $b$ is a bit of the file $f_{d_{k_b}}$ but not cached by user $k_b$. Moreover, the union of
the intended user $k_b$ and its cover set $\mathcal{T}_b$  is called the cooperative set $\mathcal{S}_b$ of $b$, i.e. $\mathcal{S}_b=\mathcal{T}_b\cup\{k_b\}$.
\end{definition}

\begin{example}\label{example1}
Let $N=K=5$, $M=\frac{5}{2}$, $F=4$ as Fig. \ref{cache} shows. By convenience, the files are denoted by $A,B,C,D,E$ respectively, i.e, $A=\{a_1,a_2,a_3,a_4\}$ and so on.
Assume that the files have uniform popularity and the original placement randomly selects two bits from each file, as Fig. \ref{fig} illustrates.
Suppose that the users 1,2,3,4,5 request files $A,B,C,D,E$ respectively. Table \ref{table1} gives the intended user, cover set and cooperative set of those bits that are needed to be transmitted.
\end{example}

\begin{figure}[htbp]
\centering\includegraphics[width=0.32\textwidth]{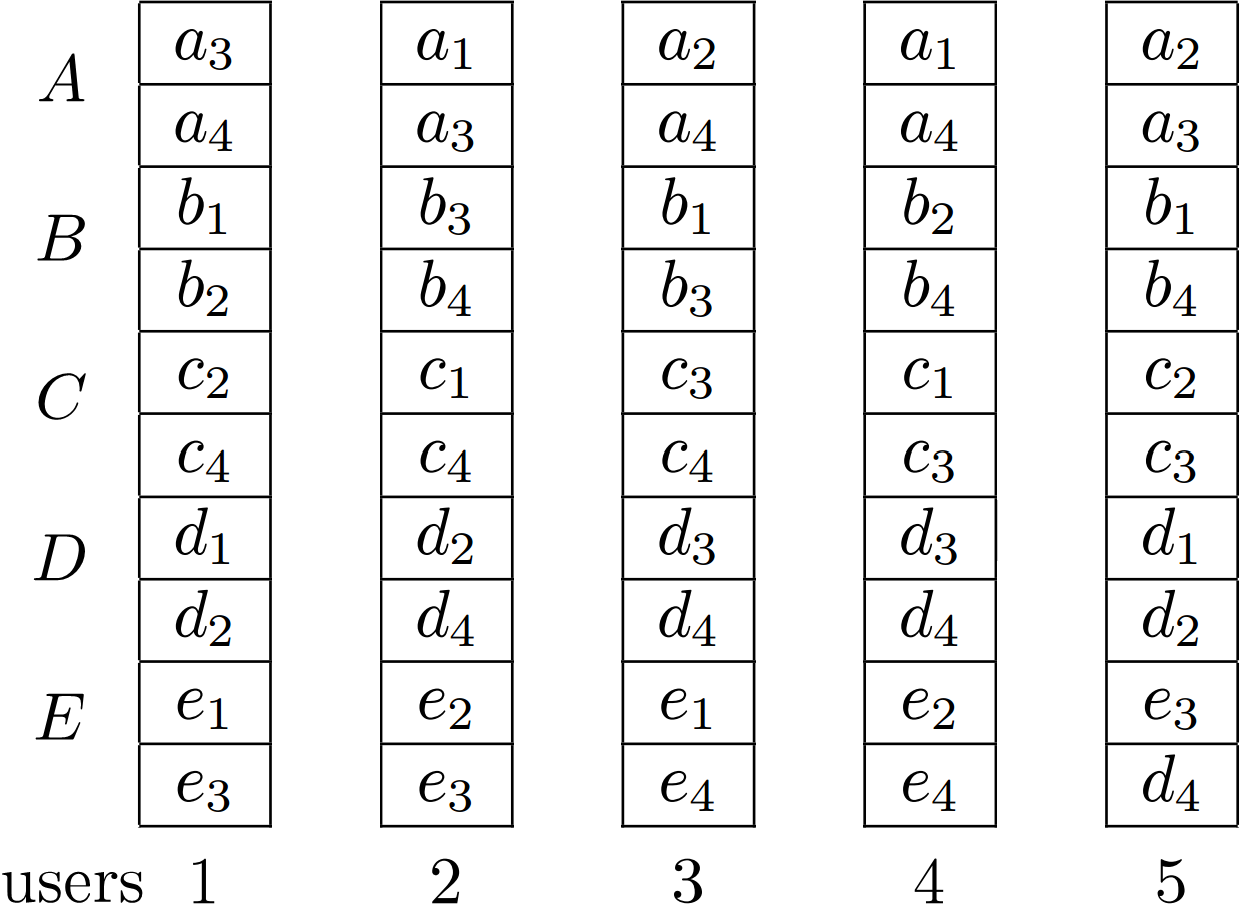}
\caption{ Contents of users' cache in Example \ref{example1}: each user randomly caches two bits of each file.}\label{fig}
\end{figure}

   \begin{table}[!htp]
  \centering
  \caption{Cover property of bits}\label{table1}
  \begin{tabular}{ccccc}
    \toprule
   Bit&Intended user  & Cover set  & Cooperative set  \\
    \midrule
 $a_1$&\multirow{2}{0.03\textwidth}{1}&$\{2,4\}$&$\{1,2,4\}$\\
 $a_2$&&$\{3,5\}$&$\{1,3,5\}$\\
 \midrule
 $b_1$&\multirow{2}{0.03\textwidth}{2}&$\{1,3,5\}$&$\{1,2,3,5\}$\\
 $b_2$&&$\{1,4\}$&$\{1,2,4\}$\\
 \midrule
 $c_1$&\multirow{2}{0.03\textwidth}{3}&$\{2,4\}$&$\{2,3,4\}$\\
 $c_2$&&$\{1,5\}$&$\{1,3,5\}$\\
 \midrule
  $d_1$&\multirow{2}{0.03\textwidth}{4}&$\{1,5\}$&$\{1,4,5\}$\\
 $d_2$&&$\{1,2,5\}$&$\{1,2,4,5\}$\\
 \midrule
  $e_1$&\multirow{2}{0.03\textwidth}{5}&$\{1,3\}$&$\{1,3,5\}$\\
 $e_2$&&$\{2,4\}$&$\{2,4,5\}$\\
 \bottomrule
  \end{tabular}
\end{table}

\subsection{Original Coded Delivery Algorithm}\label{review-sec}

Algorithm \ref{Ori_delivery} depicts the \emph{original coded delivery}  (OD) in \cite{maddah2013decentralized},
which visits all subsets of $\mathcal{K}$ in decreasing order, and transmits the bits only when their  cooperative set is exactly the current visited set.
\begin{algorithm}[!ht]
\caption{Original Coded Delivery}\label{Ori_delivery}
\begin{algorithmic}[1]
\Procedure {Original Delivery}{$d_1,d_2,\cdots,d_K$}
\For{$s=K,K-1,\cdots,1$}
\For{$\mathcal{S}\subset \mathcal{K},|\mathcal{S}|=s$}
\State Server sends $\oplus_{k\in \mathcal{S}}V_{k,\mathcal{S}\backslash\{k\}}$;\label{alg_VkS}
\EndFor
\EndFor
\EndProcedure
\end{algorithmic}
\end{algorithm}

In Algorithm \ref{Ori_delivery}, $V_{k,\mathcal{S}\backslash\{k\}}$  consists of bits intended for user $k$ whose cooperative set is $\mathcal{S}$.
It was claimed in  \cite{maddah2013decentralized} that each $V_{k,\mathcal{S}\backslash\{k\}}$ is assumed to be zero padded to the length of the longest one
for a given cooperative set $\mathcal{S}$.
This is to say, Line 4 in Algorithm \ref{Ori_delivery}
\begin{flushleft}
{\small 4}:~~~Server sends $\oplus_{k\in \mathcal{S}}V_{k,\mathcal{S}\backslash\{k\}}$;
\end{flushleft}
is equivalent to
\begin{flushleft}
{\small 4-1}:~~~$l_{\mathcal{S}}=\max_{k\in \mathcal{S}} |U_{k,\mathcal{S}\backslash\{k\}}|$;\label{min}\label{greedy2}\\
{\small 4-2}:~~~\textbf{if}~~~$l_{S}>0$ ~~\textbf{then}\\
{\small 4-3}:~~~~~~~Server sends $\oplus_{k\in \mathcal{S}}V_{k,\mathcal{S}\backslash\{k\}}$;\\
{\small 4-4}:~~~\textbf{end if}
\end{flushleft}
where $U_{k,\mathcal{S}\backslash\{k\}}$ represents the bits intended for user $k$ whose cooperative set is $\mathcal{S}$ and $V_{k,\mathcal{S}\backslash\{k\}}$ is formed by padding $U_{k,\mathcal{S}\backslash\{k\}}$ to the length $l_{\mathcal{S}}$.

For the request in Example \ref{example1}, Table \ref{table2} describes the transmitting steps of OD, where only the nonempty sets visited are listed for simplicity.
\begin{center}
  \tabcaption{Transmitting steps of OD}\label{table2}
\begin{tabular}{ccc}
    \toprule
   Time slot&Intended user  & Transmitted signal    \\
    \midrule
 1&$\{1,2,3,5\}$& $0\oplus b_1\oplus 0\oplus 0$\\
 2&$\{1,2,4,5\}$&$0\oplus 0\oplus  d_2\oplus 0$\\
 3&$\{1,2,4\}$&$a_1\oplus b_2\oplus 0$\\
 4&$\{1,3,5\}$&$a_2\oplus c_2\oplus e_1$\\
 5&$\{1,4,5\}$&$0\oplus d_1\oplus 0$\\
 6&$\{2,3,4\}$&$0\oplus c_1\oplus 0$\\
 7&$\{2,4,5\}$&$0\oplus0\oplus e_2$\\
 \midrule
 &Total time slots:& 7\\
 \bottomrule
  \end{tabular}
\end{center}

We notice that when  OD visits $\mathcal{S}=\{1,2,4,5\}$, $l_{\mathcal{S}}=1$ since (I) $d_2$ is intended for user $4$ with the cooperative set being
\emph{exactly} $\mathcal{S}$; (II) there are no bits for users $1,2,5$ whose cooperative set are $\mathcal{S}$.
Therefore, the OD   transmits $0\oplus 0\oplus d_2\oplus 0$ where three ``$0$"s are padded for users $1,2,5$ respectively.
 Then, the bit $d_2$ is indeed transmitted in $1$-multicast, though its cooperative set contains $4$ users.
In fact, $d_2$ can be left and sent later together with other two bits
$a_1$ and $b_2$ when the set $\{1,2,4\}$ is visited  since both $a_1$ and $b_2$ have the cooperative set $\{1,2,4\}$ which is contained in the
cooperative set $\{1,2,4,5\}$ of the bit $d_2$. As a result, $a_1,b_2$ and $d_2$ all gain larger multicast.

This illustrative example motivates us to propose a new coded delivery algorithm.

\subsection{Set Centered Greedy Coded Delivery Algorithm }\label{greedy-sec}

Our new coded delivery algorithm, \emph{set centered greedy coded delivery} (SGD), is given in Algorithm \ref{alg_delivery}, where
the set $U_{k,\mathcal{S}\backslash\{k\}}$ is formed by all the bits that (I) are intended for user $k$;
(II) have not been transmitted before; (III) can cover $\mathcal{S}\backslash\{k\}$,
and the set $W_{k,\mathcal{S}\backslash\{k\}}$ consists of the first $l_\mathcal{S}$ bits in  $U_{k,\mathcal{S}\backslash\{k\}}$,
for each user $k$  in the current set $\mathcal{S}$.

\begin{algorithm}[htb]
\caption{Set Centered Greedy Coded Delivery}\label{alg_delivery}
\begin{algorithmic}[1]
\Procedure {Set Greedy Delivery}{$d_1,d_2,\cdots,d_K$}
\For{$s=K,K-1,\cdots,1$}
\For{$\mathcal{S}\subset \mathcal{K},|\mathcal{S}|=s$}
\State $l_{\mathcal{S}}=\min_{k\in\mathcal{S}} |U_{k,\mathcal{S}\backslash\{k\}}|$;\label{min}\label{greedy2}
\If{$l_{S}>0$}
\State  Server sends $\oplus_{k\in \mathcal{S}}W_{k,\mathcal{S}\backslash\{k\}}$\label{UkS};
\EndIf
\EndFor
\EndFor
\EndProcedure
\end{algorithmic}
\end{algorithm}

In contrast to the OD algorithm, the key principles of SGD algorithm are that
\begin{itemize}
\item [P3.] When the algorithm visits some $\mathcal{S}\subset \mathcal{K}$, it requires that all users in $\mathcal{S}$ have data to transmit;
\item [P4.] If any one has no data to cooperate with other users, the algorithm simply visits the next set;
\item [P5.] The remaining bits of other users in $\mathcal{S}$ will be left and visited again when the algorithm visits some subset of $\mathcal{S}$.
\end{itemize}
Based on these principles, $U_{k,\mathcal{S}\backslash\{k\}}$ is changed so that the bits' cover set is a \emph{superset} of  $\mathcal{S}\backslash\{k\}$, instead of
$\mathcal{S}\backslash\{k\}$ \emph{exactly}. Furthermore, the operator ``max''  is replaced by ``min'' so that each user in $\mathcal{S}$ actually has one bit for each transmission.
As a result, SGD exploits full cooperative opportunities for the current visited set $\mathcal{S}$. In this sense,
we say that it is greedy.

Now, a bit has a large cooperative set $\mathcal{S}$ but not fully cooperated in OD can  cooperate with some other bits whose cooperative set is \emph{not exactly} $\mathcal{S}$ but its subset.

\begin{example}\label{example2}
Given the placement in  Example \ref{example1}, Table \ref{table3} demonstrates the transmitting process of SGD,
where $d_2$ and $a_1,b_2$ are transmitted together.

\begin{center}
  \tabcaption{Transmitting steps of SGD}\label{table3}
  \begin{tabular}{ccc}
    \toprule
   Time slot&Intended user  & Transmitted signal    \\
    \midrule
 1&$\{1,2,4\}$& $a_1\oplus b_2\oplus d_2$\\
 2&$\{1,3,5\}$&$a_2\oplus c_2\oplus e_1$\\
 3&$\{2,3\}$&$b_1\oplus c_1$\\
 4&$\{4,5\}$&$d_1\oplus e_2$\\
 \midrule
 &Total time slots: &4\\
 \bottomrule
  \end{tabular}
\end{center}
\end{example}

Due to the fact that, SGD is a greedy and heuristic algorithm, it is not easy to analyze its  performance in general, but the following  theorem shows that the performance of Algorithm \ref{alg_delivery} can approach the lower bound when the cache is evenly allocated to the files. The proof is left in Appendix \ref{appendix3}.

\begin{theorem}\label{thm:achieve} When $q_i=\frac{M}{N}$, $i=1,2,\cdots,N$, and $F\rightarrow\infty$, the rate of Algorithm \ref{alg_delivery} can arbitrarily approach the lower bound in \eqref{ERa1}.
\end{theorem}

Note that both  OD and SGD visit all subsets of $\mathcal{K}$ in a decreasing order. Since
the cardinality of nonempty subset of $\mathcal{K}$ is $2^{K}-1$, the complexities of OD and SGD
are increasing exponentially with $\mathcal{K}$. Typically, when $K$ is large it is hard to implement OD and SGD. This motivates us to present another algorithm with a lower computational complexity especially when $F\ll 2^K$.

\subsection{Bit Centered Greedy Coded Delivery Algorithm}\label{sec_bitdelivery}
In this section, we propose an algorithm termed as \emph{bit centered greedy coded delivery} (BGD). Unlike
OD and SGD, BGD centers on bit rather than set by means of a technique called bit merging.

\begin{definition}
Let us represent a requested bit $b$ by $(\{b\};\{k_b\};\mathcal{T}_b)$, where $k_b$ is its intended user and $\mathcal{T}_b$ is its cover set.
We say that two bits $(b;k_{b};\mathcal{T}_{b})$\footnote{For notational simplicity, for the unmerged bits, we abbreviate $(\{b\};\{k_b\};\mathcal{T}_b)$ as $(b;k_b;\mathcal{T}_b)$.} and $(c;k_{c};\mathcal{T}_{c})$ can be \emph{merged} if $k_{b}\in\mathcal{T}_{c}$ and $k_c\in\mathcal{T}_b$.
The new merged bit is denoted by $(\{b,c\};\{k_b,k_c\};\mathcal{T}_b\cap\mathcal{T}_c)$, where $\mathcal{T}_b\cap\mathcal{T}_c$
is called cover set of the  new merged bit.
Moreover, we say that the merged bits  $(B;\mathcal{K}_B;\mathcal{T}_B)$ and  $(C;\mathcal{K}_C;\mathcal{T}_C)$ can be further merged into
$(B\cup C;\mathcal{K}_B\cup\mathcal{K}_C;\mathcal{T}_B\cap\mathcal{T}_C)$ if $\mathcal{K}_B\subset\mathcal{T}_C$ and $\mathcal{K}_C\subset\mathcal{T}_B$.
\end{definition}

Indeed,  if $(b;k_b;\mathcal{T}_b)$, $(c;k_c;\mathcal{T}_c)$ are ``mergeable", the two bits $b$ and $c$ can be transmitted in a multicast without interfering each other.
For example, in Table \ref{table1}, $b_2$ and $d_2$, represented by $(b_2;2;\{1,4\})$ and $(d_2;4;\{1,2,5\})$ respectively, can be sent in a multicast.  Then,
they are merged into $(\{b_2,d_2\};\{2,4\};\{1\})$, where $\{1\}=\{1,4\}\cap\{1,2,5\}$ is covered by both bits $b_2$ and $d_2$. If a new bit can be transmitted in the same multicast, it has to cover $\{2,4\}$ and its intended user has to be in $\{1\}$. In Table \ref{table1}, $a_1$,  represented by $(a_1;1;\{2,4\})$, is the only bit satisfying this requirement. Therefore, the new merged bit $(\{b_2,d_2,a_1\};\{1,2,4\};\emptyset)$ is formed.

In BGD, the requested bits are sorted according to the sizes of their cooperative sets in decreasing order\footnote{If the size of two subset of $\mathcal{K}$ are equal, they are sorted according to  lexicographic order, see table \ref{table4}.}. Next, the algorithm maintains a list including all the  untransmitted bits and visits the bits in the list
 one by one. When a bit is visited,  it tries to find a bit in the list to be merged  such that the resultant merged cover set maintains the  largest size.  Then, the merged bit   continues to search for a new bit to merge until  no bit can be merged further in the list or the cover set become empty set. Algorithm \ref{bitdelivery} formally states the BGD procedure, which is demonstrated  with Example \ref{example3}.
\begin{algorithm}[htb]
\caption{Bit Centered Greedy Coded Delivery}\label{bitdelivery}
\begin{algorithmic}[1]
\Procedure {Bit Greedy Delivery}{$d_1,d_2,\cdots,d_K$}
\State Establish a bit list, $List=\cup_{k\in\mathcal{K}}f_{d_k}\backslash Z_{k,f_{d_k}}$, where  $Z_{k,f_{d_k}}$ is the partial part of file $f_{d_k}$ cached by user $k$. Then sort the bits in $List$ in decreasing order according to their cooperative set.

\For{$b\in List$}
 \State $B=\{b\},\mathcal{U}=\{k_b\},\mathcal{T}=\mathcal{T}_b$;
\State $Ls=List\cap \{c:k_c\in\mathcal{T}, \mathcal{U}\subset\mathcal{T}_c\}$;
\While {$\mathcal{T}\neq \emptyset~ \mbox{and}~ Ls\neq\emptyset$}
\State $l_{max}=\max_{c\in Ls}\{|\mathcal{T}_c\cap\mathcal{T}|\}$;
\State$\hat{c}=$ the last $c \in Ls$ such that $|\mathcal{T}_c\cap\mathcal{T}|=l_{max}$;
\State $B=B\cup\{\hat c\}$;
\State $\mathcal{U}=\mathcal{U}\cup \{k_{\hat{c}}\}$;
\State $\mathcal{T}=\mathcal{T}\cap\mathcal{T}_{\hat{c}}$;
\State $Ls=Ls\cap\{c:c\in \mathcal{T},\mathcal{U}\subset\mathcal{T}_c\}$;
\EndWhile
\State Transmit $\oplus_{v\in B} v$;
\State Delete the transmitted bits in $List$.
\EndFor
\EndProcedure
\end{algorithmic}
\end{algorithm}

\begin{example}\label{example3}
For a fair comparison,  Table \ref{table4} illustrates the transmitting process of BGD for the placement in  Example \ref{example1}. The initial  $List$ is given by the first column of Table \ref{table4} where the bits are sorted according to their cooperative sets. BGD visits $List$ one by one. It starts by visiting $b_1$. Hence the initial value of $B,\mathcal{U},\mathcal{T}, Ls$ are
\begin{align}
 B=\{b_1\},~\mathcal{U}=\{2\},~\mathcal{T}=\{1,3,5\},~Ls=\{a_1,c_1,e_2\}.\notag
\end{align}
In the first execution of the while loop, since the cover sets of $a_1,c_1,e_2$ are all $\{2,4\}$, and $\mathcal{T}\cap\{2,4\}=\emptyset$, $l_{max}=0$. Thus the last bit $e_2$ in $Ls$ is chosen to form a new merged bit. Thus $B,\mathcal{U},\mathcal{T},Ls$ are renewed as
\begin{align}
 B=\{b_1,e_2\},~\mathcal{U}=\{2,5\},~\mathcal{T}=\emptyset,~Ls=\emptyset,\notag
\end{align}
which finishes the while loop.  Then BGD transmits $b_1\oplus e_2$, and deletes $\{b_1,e_2\}$ from $List$. Subsequently, BGD continues to visit the next bit $d_2\cdots$, until the $List$ becomes empty.
  \end{example}

The reason for BGD choosing the \emph{last} bit that leads to the largest cover set size is that, since  we have sorted the bits according their cooperative set, the bit in the later position of $List$ tends to have a smaller cover set size, i.e., poorer cover property. So, BGD gives priority to sending such bit if it can provide the same size of merged cover set.

  \begin{table*}[!htp] \centering  \caption{Transmitting steps of BGD} \label{table4}
 \begin{tabular}{cccccc} \toprule
 Bit list&Cooperative set&Representation&Merged bit&Transmitted signal&Time slot\\ \midrule
 $b_1$&$\{1,2,3,5\}$&$(b_1;2;\{1,3,5\})$&$(\{b_1,e_2\};\{2,5\};\emptyset)$& $b_1\oplus e_2$&1\\
 $d_2$&$\{1,2,4,5\}$&$(d_2;4;\{1,2,5\})$&$(\{d_2,a_1,b_2\};\{4,1,2\};\emptyset)$&$d_2\oplus a_1\oplus b_2$&2\\
  $a_1$&$\{1,2,4\}$&$(a_1;1;\{2,4\})$&&&\\
   $b_2$&$\{1,2,4\}$&$(b_2;2;\{1,4\})$&&&\\
     $a_2$&$\{1,3,5\}$&$(a_2;1;\{3,5\})$&$(\{a_2,c_2,e_1\};\{1,3,5\};\emptyset)$&$a_2\oplus c_2\oplus e_1$&3\\
      $c_2$&$\{1,3,5\}$&$(c_2;3;\{1,5\})$&&&\\
       $e_1$&$\{1,3,5\}$&$(e_1;5;\{1,3\})$&&&\\
       $d_1$&$\{1,4,5\}$&$(d_1;4;\{1,5\})$&$(d_1;4;\{1,5\})$&$d_1$&4\\
        $c_1$&$\{2,3,4\}$&$(c_1;3;\{2,4\})$&$(c_1;3;\{2,4\})$&$c_1$&5\\
         $e_2$&$\{2,4,5\}$&$(e_2;5;\{2,4\})$&&&\\
\midrule
   &&&&Total time slots:&5\\
      \bottomrule
 \end{tabular} \end{table*}

Note that SGD visits \emph{every} subset of $\mathcal{K}$ in decreasing order to utilize the cooperative opportunities. Unlike SGD, the while loop in BGD searches \emph{part} of subsets of $\mathcal{K}$ according to the bits' cover property. Since it  seeks smaller range compared to SGD, its performance is slightly inferior to SGD in  most cases in general. But this is not always true, Example \ref{example4}  shows a case that SGD consumes more time slots than BGD.

\begin{example}\label{example4}
We revise Example \ref{example1} slightly:  exchange $d_3$ in the cache of user $3$ with $d_1$ in the cache of user $5$. Accordingly, the cover set and cooperative set of $d_1$ (Table \ref{table1}) is revised to  be $\{1,3\}$ and $\{1,3,4\}$ respectively. Through executing SGD and BGD for such a revised version  in Table \ref{table5}
it is verified that SGD consumes 5 time slots while BGD consumes 4 time slots.
\begin{center}
  \tabcaption{Transmitting steps in Example \ref{example4}}\label{table5}
  \begin{tabular}{ccc}
    \toprule
   Time slot&SGD  & BGD    \\
    \midrule
 1&$a_1\oplus b_2\oplus d_2$&$b_1\oplus e_2$\\
 2&$a_2\oplus c_2\oplus e_1$&$d_2\oplus a_1\oplus b_2$\\
 3&$b_1\oplus c_1$&$d_1\oplus c_1$\\
 4&$d_1$&$a_2\oplus c_2\oplus e_1$\\
 5&$e_2$&\\
 \midrule
 Total time slots:&5 &4\\
 \bottomrule
  \end{tabular}
\end{center}
\end{example}

It should be noted that the idea of merging bits has been employed in  a delay-sensitive framework \cite{NiesenDelay}. In BGD, two ingredients are new:
(I) The bits are sorted according to their cooperative sets before transmitting. By doing this, the bits with larger cooperative sets are visited firstly and tends to be sent in a multicast with more users; (II) Instead of using a \emph{misfit} function,
\begin{align}
&\rho\left((B;\mathcal{K}_B;\mathcal{T}_B),(C;\mathcal{K}_C;\mathcal{T}_C)\right)\notag\\
\overset{\triangle}{=}&|\mathcal{T}_B\backslash(\mathcal{T}_C\cup\mathcal{K}_C)|+|\mathcal{T}_C\backslash(\mathcal{T}_B\cup\mathcal{K}_B)|\notag
\end{align}
 BGD uses a \emph{fit} function,
\begin{align}
\eta\left((B;\mathcal{K}_B;\mathcal{T}_B),(C;\mathcal{K}_C;\mathcal{T}_C)\right)\overset{\triangle}{=}|\mathcal{T}_B\cap\mathcal{T}_C|\notag
\end{align}
to measure the fitness of merging two mergeable bits. For two mergeable bits, the misfit function $\rho(\cdot,\cdot)$ counts the number of the users wasted due to merging the two bits, while the fit function $\eta(\cdot,\cdot)$ counts the number of the users we may further accept in the future. Since we are pursuing multicasting,  $\eta(\cdot,\cdot)$ is preferred in our setting.

Note that, SGD consumes $2^K-1$ steps to visit the nonempty subsets of $\mathcal{K}$ sequentially, while  the BGD can be accomplished in $K(1-\frac{M}{N})F$ steps. Thus, it is obvious that whenever $K$ is large such that $F\ll 2^K$, BGD greatly reduces the complexity.

In summary, SGD and BGD exploit more cooperative opportunities between the  bits with different cooperative sets, so attains a better performance. Particularly, BGD is designed to decrease the complexity in circumstances when  $K$ is large and $F\ll 2^K$.

\section{Simulation}\label{sec_simu}

In this section, we conduct  simulations to compare  the \emph{new placement} (NP) proposed in Section \ref{allocation} with OP, GP and SGD, BGD in Section \ref{delivery} with  OD, GD. 

To characterize the popularity of different files,
Zipf distribution \cite{Zipf2007} is employed in our simulation. In Zipf distribution, the popularity of the $i$th most popular file is
\begin{align}
p_i=\frac{\frac{1}{i^\alpha}}{\sum_{j=1}^N\frac{1}{j^\alpha}}\notag
\end{align}
where $N$ is the number of files and $\alpha>0$ is the decay parameter. A larger $\alpha$ indicates that a small number of files accounts for the majority of traffic, while a smaller $\alpha$ indicates a more uniform distribution of popularity. Particularly, $\alpha=0$ indicates uniform popularity distribution.

Recall that, SGD simply changes ``max" in Line 4-1 (Section \ref{review-sec}) in OD to ``min" to avoid padding zeros to bit vectors.  To illustrate the effect of padding zeros to the performance, in the simulation, we add an additional algorithm, i.e., Semi-Set Centered Greedy Coded Delivery (Semi-SGD), which is depicted in Algorithm \ref{alg_semiSGD}.

\begin{algorithm}[htb]
\caption{Semi-Set Centered Greedy Coded Delivery}\label{alg_semiSGD}
\begin{algorithmic}[1]
\Procedure {Semi-Set Greedy Delivery}{$d_1,d_2,\cdots,d_K$}
\For{$s=K,K-1,\cdots,1$}
\For{$\mathcal{S}\subset \mathcal{K},|\mathcal{S}|=s$}
\State $l_{\mathcal{S}}=\left\lfloor\left(\min_{k\in\mathcal{S}} |U_{k,\mathcal{S}\backslash\{k\}}|+\max_{k\in\mathcal{S}\backslash\{k\}}|U_{k,\mathcal{S}\backslash\{k\}}|\right)/2\right\rfloor$;\label{min}\label{greedy2}
\If{$l_{S}>0$}
\State  Server sends $\oplus_{k\in \mathcal{S}}W_{k,\mathcal{S}\backslash\{k\}}$\label{UkS};
\EndIf
\EndFor
\EndFor
\EndProcedure
\end{algorithmic}
\end{algorithm}

In Algorithm  \ref{alg_semiSGD}, $U_{k,\mathcal{S}\backslash\{k\}}$ are formed in the same way as that in SGD. If $|U_{k,\mathcal{S}\backslash\{k\}}|\geq l_{\mathcal{S}}$, the set $W_{k,\mathcal{S}\backslash\{k\}}$ consists of the first $l_\mathcal{S}$ bits in  $U_{k,\mathcal{S}\backslash\{k\}}$; or else, $W_{k,\mathcal{S}\backslash\{k\}}$  is formed by padding zeros on $U_{k,\mathcal{S}\backslash\{k\}}$ to the length $l_\mathcal{S}$. Intuitively, it is a intermediate case between OD and SGD.

For a fair comparison, we mainly focus on two cases: $F\gg2^K$ and $F\ll2^K$ for both uniform and non uniform popularity. $F\gg 2^K$ indicates a large enough $F$, which is consistent with the theoretical assumptions of \cite{maddah2013decentralized} and \cite{niesen2013coded}, while $F\ll2^K$ is
a more practical case when $K$ is  large. All of our results are averaged over 5000 simulations.
\subsection{Uniform Popularity}
In this case, $\alpha=0$. Moreover, GP and GD are the same with OP and OD respectively. NP also leads to the OP as well. Thus, we only need to compare OD, SGD, BGD and Semi-SGD  under  OP. For $F\gg2^K$, we set $K=8,F=10000$ while for case $F\ll 2^K$, we set $K=16, F=1000$. In both cases, we set $N=100$. For comparison, we also plot the  curve of uncoded delivery. In this scheme, the server sends the required bits to each user one by one, which achieves a rate
\begin{eqnarray*}
R_U(M)=K\left(1-\frac{M}{N}\right).
\end{eqnarray*}

Fig. \ref{unif8} shows the case $F\gg 2^K$. As mentioned before, when $F\rightarrow\infty$, the lower bound for OD is tight. We observe that,  even in this case,  SGD, BGD and Semi-SGD can outperform OD slightly. The performance of Semi-SGD is between OD and SGD, since the amount of padded zeros in Semi-SGD is intermediate between those in OD and SGD. 

  Fig. \ref{unif16} shows the case $F\ll 2^K$. In this case, both SGD and BGD can
outperform OD significantly, since the law of large numbers plays out.  As expected, Semi-SGD is still between SGD and OD.  In this case, OD needs large amounts of ``zeros" to be padded to the cooperative bits, which is inefficient. The zeros padded to Semi-SGD are less than OD but more than SGD.  Whereas, both SGD and BGD attain more cooperative opportunities.

\begin{figure*}[!htb]
                \centering
                \subfigure[$F\gg2^K$]{         
                \label{unif8}         
                \includegraphics[width=0.8\textwidth ]{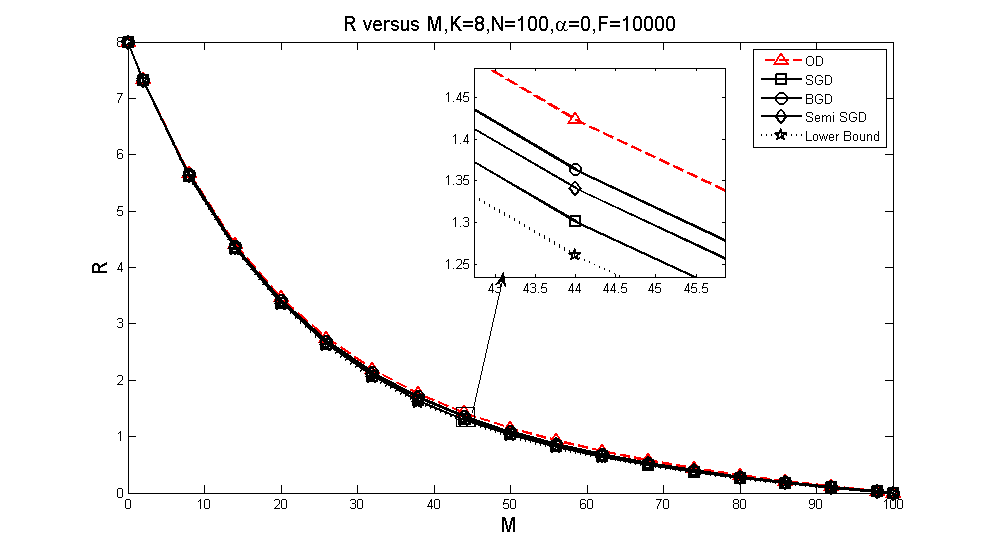}}
                \subfigure[$F\ll2^K$]{        
                \label{unif16}         
                \includegraphics[width=0.8\textwidth ]{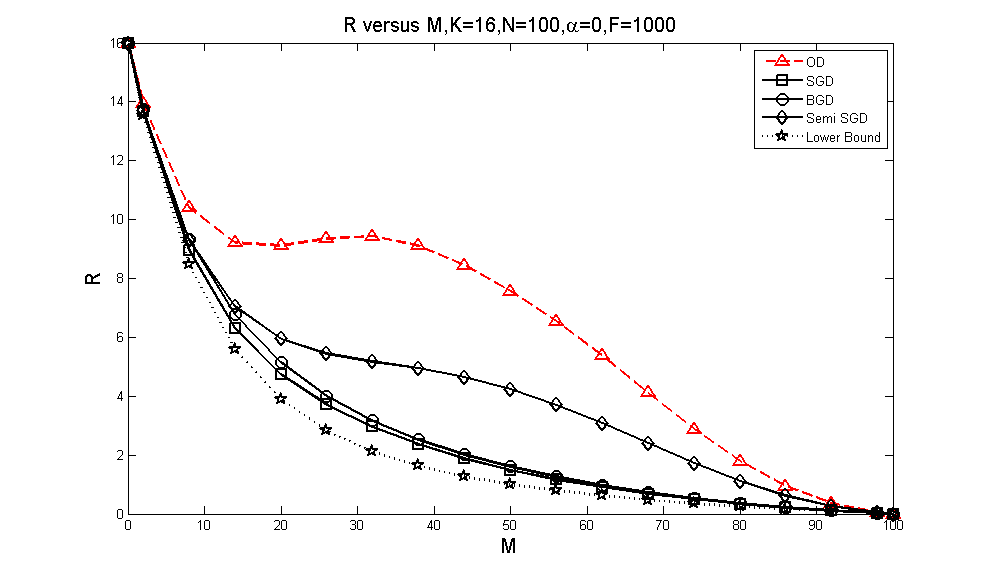}}
                \caption{Rate-memory tradeoff curves, uniform popularity}\label{fig_sim}
\end{figure*}

 The gaps between the uncoded curve and the coded curves indicate the gain from coding. For OD, when $M/N$ is around $1/2$, it is the most difficult to benefits from coding, while for small or large $M/N$, it becomes more easy to obtain coding gains. In fact, there are totally $(1-M/N)F=O(F)$ bits intended for each user. The sizes of the cooperative set of these bits are mainly around its expectation $1+(K-1)\cdot M/N\approx \lceil KM/N\rceil$. There are at least  ${K\choose \lceil KM/N\rceil}$ such sets. In the case of $M/N=1/2$, ${K\choose \lceil K/2\rceil}$ is the  largest  (the same order as $\frac{1}{\sqrt{2\pi K}}2^{K+1}$), thus the cooperative sets of bits intended for different users are very likely to be distinct (recall that $F\ll 2^K$),  which violates the cooperative condition that the cooperative sets of XORed bits should be exactly same in OD. Example 1 (Table \ref{table1}) is of this case. While in case $M/N$ is small or large, ${K\choose \lceil KM/N\rceil}$ is small compared to $O(F)$, the bits' cooperative sets are more likely to be same, which is consistent with the condition of cooperation  in OD. While since SGD and BGD can capture the cooperative opportunities between the bits with different cooperative sets, their performances are consistent with our intuition - the larger memory, the less rate.

In particular, we see that in both Fig. \ref{unif8} and \ref{unif16}, BGD's performance is just slightly inferior to SGD, although it does not search all subsets of users. The experimental results support that, BGD can achieve almost the same performance as SGD. Thus, in the case that $K$ is large, BGD is a good substitute for SGD.

\subsection{Nonuniform Popularity}
In this case, we revise $\alpha$ to be $0.6$ and keep all other parameters unchanged.  For new placement,  two versions are proposed in Section \ref{allocation}, i.e. allocating the cache according to \eqref{solution} and \eqref{solution2} respectively. For simplicity, we abbreviate the new placement according to \eqref{solution} as NP1 and the one according to \eqref{solution2} as NP2.
 Firstly, we compare the performances of OD, SGD, BGD, Semi-SGD as well as the lower bound \eqref{ERa1} under the even cache allocation placement (i.e., OP); Secondly,  we make use of  the optimal value of caching allocation parameter $M_l$ for fair comparison in GP, which is  determined in Appendix \ref{appendix2}; Thirdly, we match each of NP1 and NP2 with each of SGD, BGD and Semi-SGD. Lastly, for comparison, we also plot the lower bound \eqref{ER}, in which the parameters $q_i$ are optimized according to \eqref{solution}.

 \begin{figure*}[!htb]
                \centering
                \subfigure[$F\gg2^K$]{         
                \label{nunif8}         
                \includegraphics[width=0.8\textwidth ]{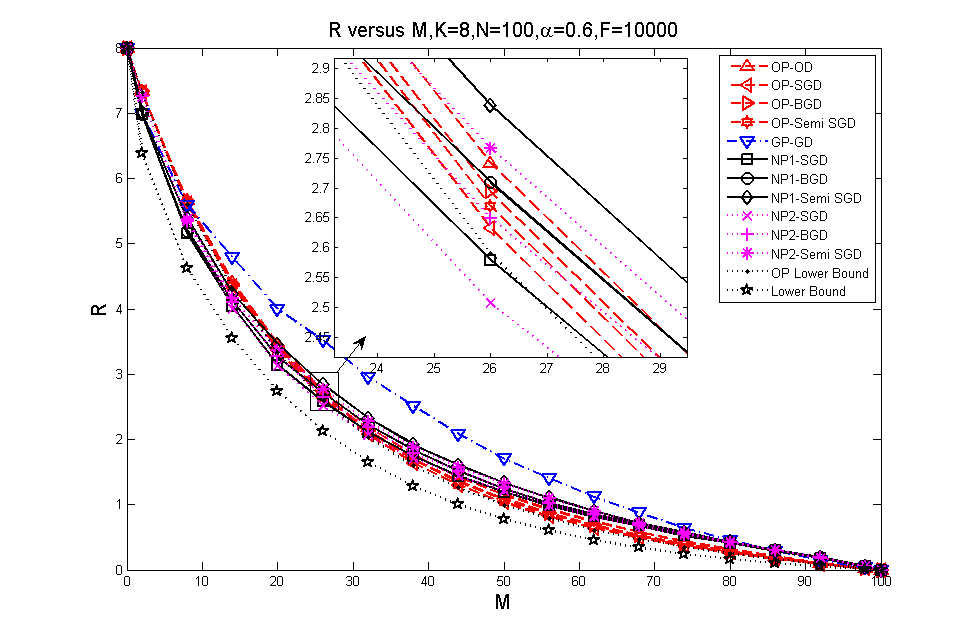}}
                \subfigure[$F\ll2^K$]{        
                \label{nunif16}         
                \includegraphics[width=0.8 \textwidth ]{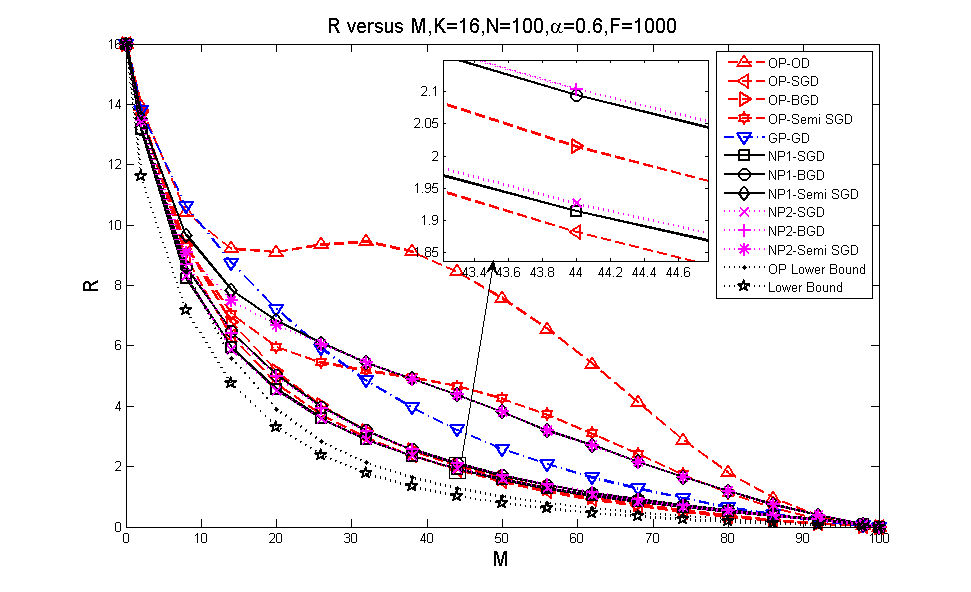}}
                \caption{Rate-memory tradeoff curve, nonuniform popularity}\label{fig_sim2}
\end{figure*}

Fig. \ref{nunif8} and \ref{nunif16} shows the results for $F\gg 2^K$ and $F\ll 2^K$ respectively. Several common observations are: (I) With a fixed delivery (SGD, BGD or Semi-SGD),  NP1 and NP2 are almost the same, this validates that, even in the case that the server has no knowledge of $K$, it can still allocates the cache as in \eqref{solution2}, for most $K$, the effect is almost the same with \eqref{solution}.  (II) With a fixed placement (NP1 or NP2), SGD achieves the best performance, BGD is slightly inferior to SGD, and the performance of Semi-SGD is inferior to SGD, which is consistent with Fig. \ref{fig_sim}.

In case $F\gg2^K$, similar to the uniform case, the performances of OP, SGD, BGD and Semi-SGD under OP are similar, i.e., they all approach the lower bound \eqref{ERa1}, while in case $F\gg 2^K$, the lose of Semi-SGD and OD become significant. Similarly, under NP1 and NP2, the performance distinctions of SGD, BGD and Semi-SGD is less than that in case $F\ll 2^K$. The reason is similar to that in Fig. \ref{fig_sim2}, i.e., the padded zeros become much more due to the lack of bits.
They all outperform GP-GD in case $F\gg 2^K$, which suffers the loss of cooperative opportunities between users in different groups. While in case $F\ll 2^K$, GP-GD is better than OD and Semi-SGD in some interval. This is because although $F$ may be too small to capture the cooperative opportunities in OD and Semi-SGD, it is possible sufficient for grouping scheme, since the number  of users in each group is usually much smaller than $K$.



In summary, SGD and BGD exploit the cooperative opportunities efficiently, especially in case $F\ll 2^K$. Semi-SGD is an intermediate algorithm between SGD and OD in term of the amount of zeros padded to the transmitting  bit vectors, and so is its performance. BGD's performance is close to SGD, and thus can be used to substitute it in some cases.

\section{Conclusions}\label{sec_con}

In this paper, we identified a generic family of caching allocation placement and a generic family of XOR cooperative delivery for decentralized coded caching system. We first derived a lower bound on the rate for the identified family, and then proposed a new placement scheme and two delivery schemes, i.e., SGD and BGD. In the new placement scheme, the server can approximately optimize the cache allocation even when it has no knowledge of users number $K$. In the SGD scheme, we relaxed the cooperative condition in the original delivery algorithm, which enables SGD to capture more cooperative opportunities. Furthermore, from a practical implementation view, SGD can achieve a significant coding gain with a reasonable large $F$.
Particularly,  to realize  coding gains when the computation burden of SGD can not be accepted in reality, BGD is designed such that it  can greatly reduce the complexity of SGD but still maintains the considerable cooperative opportunities typically when $K$ is large and $F\ll 2^K$.

\begin{appendices}
\section{}\label{appendix1}
\begin{proof}[Proof of Lemma \ref{lem1}]
When $x\rightarrow 0+$, $h(x)$ is a limit problem of type $\frac{0}{0}$. Thus we can use L'Hospital's rule to obtain the left limit of $h(x)$ at $0$:
\begin{align}
&\lim_{x\rightarrow 0+}h(x)\notag\\
&=\lim_{x\rightarrow 0+}\frac{2x}{-K(1-x)^{K-1}(-1)(1+Kx)-K(1-x)^K}\notag\\
&=\frac{2}{K}\lim_{x\rightarrow 0+}\frac{1}{(K+1)(1-x)^{K-1}}\notag\\
&=\frac{2}{K(K+1)}.\notag
\end{align}

In order to prove that $h(x)$ is monotonous on $(0,1]$, we only need to check that
\begin{align}
h'(x)>0,~\forall ~x\in(0,1),\label{h1}
\end{align}
where $h'(x)$ is given by

\begin{align}
&h'(x)\notag\\
=&\frac{2x\left(1-(1-x)^K(1+Kx)\right)-K(K+1)x(1-x)^{K-1}x^2}{\left(1-(1-x)^K(1+Kx)^2\right)^2}\notag\\
=&\frac{x}{\left(1-(1-x)^K(1+Kx)^2\right)^2} \cdot\left(2-2(1-x)^K(1+Kx)-K(K+1)x^2(1-x)^{K-1}\right).\label{h}
\end{align}

Therefore \eqref{h1} is equivalent to \eqref{h4}, \eqref{h4}, and \eqref{h3}
\begin{align}
&2-2(1-x)^K(1+Kx)-K(K+1)x^2(1-x)^{K-1}>0,~\forall x\in(0,1)\label{h4}\\
\Leftrightarrow~ &2-2y^K(1+K(1-y))-K(K+1)(1-y)^2y^{K-1}>0,~\forall y\in(0,1)\label{h2}\\
\Leftrightarrow~ &-K(K-1)y^{K+1}+2(K+1)(K-1)y^K-K(K+1)y^{K-1}+2>0,~\forall y\in(0,1),\label{h3}
\end{align}

where in \eqref{h2}, $y=1-x$. Denote $l(y)$ the left side of \eqref{h3}, then we can prove \eqref{h1} via proving
\begin{align}
l(y)>0,~\forall~ y\in(0,1).\notag
\end{align}

Now, $l'(y)$ can be computed:
\begin{align}
l'(y)&=-(K-1)K(K+1)y^K+2(K-1)K(K+1)y^{K-1}-(K-1)K(K+1)y^{K-2}\notag\\
&=(K-1)K(K+1)y^{K-2}(-y^2+2y-1)\notag\\
&=-(K-1)K(K+1)y^{K-2}(1-y)^2<0,~\forall ~y\in (0,1).\notag
\end{align}
Therefore, $l(y)$ is decreasing on $(0,1)$. Since $l(y)$ is continuous at $1$, we have
\begin{align}
l(y)>l(1)=0,~\forall ~y\in(0,1).\notag
\end{align}
That is, \eqref{h1} is proved and then $h(x)$ is increasing on $(0,1]$. Hence,
\begin{align}
h(x)<h(1)=1,~\forall~ x\in(0,1).\notag
\end{align}
\end{proof}

\begin{proof}[Proof of Lemma \ref{lem2}]
Firstly, to prove that $f(x)$ is continuous on $[0,1]$, it is sufficient to show
\begin{align}
\lim_{x\rightarrow 0+}\frac{1-x}{x}\left(1-(1-x)^K\right)=K.\notag
\end{align}
Again as a limit problem of type $\frac{0}{0}$, it can be easy proved with L'Hospital rule.

Next, for $\forall~x\in(0,1)$,
\begin{align}
f'(x)&=\frac{(1-x)^K(1+Kx)-1}{x^2}\notag\\
&=-\frac{1}{h(x)}.\notag
\end{align}
By Lemma \ref{lem1}, $h(x)$ in monotonously increasing on $(0,1)$ and $h(x)>0$ for all $x\in(0,1)$, then
$f'(x)$ is monotonously increasing on $(0,1)$ with
\begin{align}
f'(x)<0,~\forall~x\in(0,1).\notag
\end{align}
Therefore, $f(x)$ is a monotonously decreasing and convex function on $[0,1]$.
\end{proof}

\section{}\label{appendix2}
In this appendix, we solve the problem of finding the optimal cache allocation $M_l$ when the grouping strategy is fixed, i.e., we solve the problem
\begin{align}
\minimize_{M_1,M_2,\cdots,M_L}\quad&\sum_{l=1}^L \mathbb{E} \left[R(M_l,N_l,\mathsf{K}_l)\right]\notag\\
\mbox{subject to}\quad &0\leq M_l\leq N_l,~ \forall~l\in\{1,2,\cdots,L\},\label{opt2}\\
&\sum_{l=1}^LM_l=M,\notag
\end{align}
where $R(M,N,K)$ is given by \eqref{Rd}.

Firstly, let us evaluate  $\mathbb{E}\left[R(M_l,N_l,\mathsf{K}_l)\right]$. Denote $P_l=\sum_{l:f_l\in \mathcal{N}_l}p_l$, then $\mathsf{K}_l$ is the number of users whose requested files lie in $\mathcal{N}_l$. For a single user, the probability of the event that its requested file lies in $\mathcal{N}_l$ is $P_l$. Since different users' requests are assumed to be independent, $\mathsf{K}_l$ is a random variable that obeys the binomial distribution of parameters $K,P_l$, i.e.,
\begin{align}
\mathbb{P}(\mathsf{K}_l=k)={K \choose k}P_l^k(1-P_l)^{K-k},~k\in \{0,1,2,\cdots,K\}.\notag
\end{align}
Thus, $\mathbb{E}\left[R(M_l,N_l,\mathsf{K}_l)\right]$ can be derived:

\begin{align}
\mathbb{E}\left[R(M_l,N_l,\mathsf{K}_l)\right]
&=\sum_{k=0}^K\frac{1-\frac{M_l}{N_l}}{\frac{M_l}{N_l}}\left(1-\left(1-\frac{M_l}{N_l}\right)^k\right){K \choose k}P_l^k(1-P_l)^{K-k}\notag\\
&=\frac{1-\frac{M_l}{N_l}}{\frac{M_l}{N_l}}\left[\sum_{k=0}^K{K \choose k}P_l^k(1-P_l)^{K-k}-\sum_{k=0}^K{K \choose k}\left(\left(1-\frac{M_l}{N_l}\right)P_l\right)^k(1-P_l)^{K-k}\right]\label{bin1}\\
&=\frac{1-\frac{M_l}{N_l}}{\frac{M_l}{N_l}}\left(1-\left(1-\frac{M_l}{N_l}P_l\right)^K\right),\notag
\end{align}
where we use binomial formula for both the first term and the second term in the brackets of \eqref{bin1}.

Let $x_l=\frac{M_l}{N_l}P_l$, then the optimization problem \eqref{opt2} can be transformed into an equivalent form in terms of variable $x_l$:
\begin{align}
\minimize_{x_1,x_2,\cdots,x_L}\quad&\sum_{l=1}^L\frac{P_l-x_l}{x_l}\left(1-\left(1-x_l\right)^K\right)\label{opt3}\\
\mbox{subject to}\quad&0\leq x_l\leq P_l,\notag\\
&\sum_{l=1}^L\frac{N_l}{P_l}x_l=M.\label{eqM2}
\end{align}

The problem can be solved in a similar way as the problem \eqref{opt}. The following lemma plays a basic role:

\begin{lemma}\label{lem3}
For $0<P\leq1$, the function $$f(x)=\left\{\begin{array}{ll}
                                       \frac{P-x}{x}\left(1-(1-x)^K\right), & 0<x\leq P \\
                                       KP ,& x=0
                                     \end{array}\right.
$$ is a continuous, monotonous, decreasing and convex function on $[0,P]$. And
\begin{align}
f'(x)=\frac{\left[(1-x)^{K-1}\left((K-1)x+1\right)-1\right]P}{x^2}&-K(1-x)^{K-1},~\forall x\in (0,P),\notag
\end{align}
furthermore, if $K\geq2$, $f(x)$ is a strict convex function. When $x\rightarrow0+$,
\begin{align}
\lim_{x\rightarrow0+}f'(x)=-\frac{K(K-1)}{2}P-K.\notag
\end{align}
\end{lemma}

\begin{proof}
To prove that $f(x)$ is continuous, it suffices  to show that $\lim_{x\rightarrow 0+}f(x)=KP$, which is easily verified by L'Hospital rule.

The form of $f'(x)$ can be computed with a little complicated calculation.
If $K=1$, $f(x)=P-x$  is a monotonous decreasing convex function.
If $K\geq2$, then $K-1\geq1$, by Lemma \ref{lem1}, the first term of $f'(x)$ is increasing and less than $0$ on $(0,P)$, and the second term $-K(1-x)^{K-1}$ is also obviously increasing and less than $0$ on $(0,P)$, thus $f'(x)$ is increasing and less than $0$ on $(0,P)$. Therefore, $f(x)$ is a monotonous decreasing convex function on $[0,P]$. Furthermore, if $K\geq2$, $f(x)$ is a strict convex function on $[0,P]$. The last statement can be obtained from Lemma \ref{lem1}, which results in
\begin{align}
&\quad\lim_{x\rightarrow0+}\frac{\left[(1-x)^{K-1}\left((K-1)x+1\right)-1\right]}{x^2}\notag\\
&=\lim_{x\rightarrow0+}-\frac{1}{\frac{x^2}{1-(1-x)^{K-1}\left((K-1)x+1\right)}}\notag\\
&=-\frac{1}{\frac{2}{(K-1)K}}\notag\\
&=-\frac{K(K-1)}{2}.\notag
\end{align}
\end{proof}

Lemma \ref{lem3} implies that the optimization problem \eqref{opt3} is a convex optimization problem, whose Lagrange function is
\begin{align}
&L(x_1,\cdots,x_L,\lambda_1,\cdots\lambda_L,\delta_1,\cdots,\delta_L,\nu)\notag\\
=&\sum_{l=1}^L\frac{P_l-x_l}{x_l}\left(1-\left(1-x_l\right)^K\right)+\sum_{l=1}^L\lambda_l(x_l-P_l)-\sum_{l=1}^L\delta_lx_l+\nu\left(\sum_{l=1}^L\frac{N_l}{P_l}x_l-M\right),
\end{align}
where $\lambda_1,\cdots,\lambda_L,\delta_1,\cdots,\delta_L,\nu$ are dual variables. By the KKT conditions \cite{boyd2009convex}, we have

\begin{enumerate}
  \item If $0<x_l<P_l$, $P_l$ is given by
  \begin{align}
  P_l=&\frac{x_l^2}{1-(1-x_l)^{K-1}[(K-1)x_l+1]}
  \cdot\left[\nu\frac{N_l}{P_l}-K(1-x_l)^{K-1}\right].\label{P1}
  \end{align}

  \item If $x_l=P_l$, $P_l$  satisfies
  \begin{align}
  P_l=&\frac{P_l^2}{1-(1-P_l)^{K-1}\left[(K-1)P_l+1\right]}\cdot \left[\nu\frac{N_l}{P_l}+\lambda_l-K(1-P_l)^{K-1}\right]\notag\\
  \geq& \frac{P_l^2}{1-(1-P_l)^{K-1}\left[(K-1)P_l+1\right]}\cdot\left[\nu\frac{N_l}{P_l}-K(1-P_l)^{K-1}\right].\label{P2}
  \end{align}
  \item If $x_l=0$, $P_l$ satisfies
  \begin{align}
  P_l&=\frac{2}{K(K-1)}\left(\nu\frac{N_l}{P_l}-\delta_l-K\right)\notag\\
  &\leq \frac{2}{K(K-1)}\left(\nu\frac{N_l}{P_l}-K\right).\label{P3}
\end{align}
\end{enumerate}

Similarly to Section \ref{allocation}, for $M<N$, $\nu\geq0$, thus if we define
\begin{align}
g_l(x)=\frac{x^2}{1-(1-x)^{K-1}\left[(K-1)x+1\right]}\left[\nu\frac{N_l}{P_l}-K(1-x)^{K-1}\right],\label{gl}
\end{align}
and let $a_l$ be $\frac{2}{K(K-1)}\left(\nu\frac{N_l}{P_l}-K\right)$, $b_l$ be
$$ \frac{P_l^2}{1-(1-P_l)^{K-1}\left[(K-1)P_l+1\right]}\left[\nu\frac{N_l}{P_l}-K(1-P_l)^{K-1}\right],$$
we can express $x_l$ in the following way:
\begin{enumerate}
  \item If $\nu\geq\frac{KP_l}{N_l}$, then $a_l>0$ and $\nu\frac{N_l}{P_l}-K(1-x_l)^{K-1}>0,~\forall~ x\in [0,P_l]$, thus $g_l(x)$ is increasing on $(0,P_l)$ and then $g_l^{-1}(x)$ is uniquely defined on the interval $(a_l,b_l)$
      thus
      \begin{align}
      x_l=\left\{\begin{array}{ll}
                   P_l ,& \mbox{if } P_l\geq b_l   \\
                   g_l^{-1}(P_l) ,& \mbox{if } P_l\in  \left(a_l,b_l\right)\\
                  0 ,& \mbox{if } P_l\leq   a_l
                 \end{array}
      \right..
      \end{align}
  \item If $\nu<\frac{KP_l}{N_l}$ but $\nu\geq \frac{KP_l}{N_l}(1-P_l)^{K-1}$, then \eqref{P3} is impossible since $P_l>0$ and $\nu\frac{N_l}{P_l}-K<0$. Let $x^*_l$ be the point such that $\nu\frac{N_l}{P_l}=K(1-x^*_l)^{K-1}$, then $g_l(x)$ is increasing on $(x^*_l,P_l)$, and $g^{-1}_l(x)$ is uniquely defined on the interval $(0,b_l)$. Thus
      \begin{align}
      x_l=\left\{\begin{array}{ll}
                  P_l, & \mbox{if } P_l\geq b_l  \\
                   g_l^{-1}(P_l) ,& \mbox{if } P_l\in(0,b_l)
                 \end{array}
      \right..
      \end{align}
  \item If $\nu<\frac{KP_l}{N_l}(1-P_l)^{K-1}$, \eqref{P2} always holds, and thus $x_l=P_l$.
\end{enumerate}

Summarizing these results, we obtain the following expression of $x_l$:
\begin{align}
x_l=\left\{\begin{array}{ll}
             P_l, & \mbox{if } P_l\geq b_l \\
             g_l^{-1}(P_l), & \mbox{if } P_l\in(\max\{a_l,0\},b_l) \\
             0, & \mbox{if } P_l\leq\max{ \{a_l,0\}}
           \end{array}
\right.,
\end{align}
where $g_l(x)$ is define as \eqref{gl} on the interval $(0,P_l)$ $\left(\mbox{if } \nu>\frac{KP_l}{N_l}\right)$ or $(x^*_l,P_l)$ $\left(\mbox{if } \nu\leq \frac{KP_l}{N_l}\right)$. It is easy to see that $x_l$ are non increasing functions of $\nu$. On the other hand, $\nu$ can be determined by equation \eqref{eqM2}, thus, we can find $\nu$ through a bisection method easily, and then determine the optimal $x_l$.

Once we find the optimal $x_l$, we can determine the optimal $M_l$ by the relationship
\begin{align}
M_l=\frac{N_l}{P_l}x_l.\notag
\end{align}

\section{}\label{appendix3}

\begin{proof}[Proof of of Theorem \ref{thm:achieve}]
When $q_i=\frac{M}{N}\overset{\triangle}{=}q$,  $i=1,2,\cdots,N$, consider a  particular request $\bm d=(d_1,d_2,\cdots,d_K)$. Let $V_{k,\mathcal{S}\backslash\{k\}}$ denote the set of bits in file $W_{d_k}$ such that their cover set is exactly $\mathcal{S}\backslash\{k\}$. Then when $F\rightarrow\infty$, by the law of large numbers,  the size of $V_{k,\mathcal{S}\backslash\{k\}}$ is
\begin{align}
|V_{k,\mathcal{S}\backslash\{k\}}|=Fq^{|\mathcal{S}|-1}(1-q)^{K-|\mathcal{S}|+1}+o(F), \forall~k\in\mathcal{S}.\label{eqn:Vk}
\end{align}

In Algorithm \ref{alg_delivery}, initially, $U_{k,\mathcal{K}\backslash\{k\}}=V_{k,\mathcal{K}\backslash\{k\}}$, thus
\begin{align}
l_{\mathcal{K}}=\min_{k\in\mathcal{K}}\left\{U_{k,\mathcal{K}\backslash\{k\}}\right\}=Fq^{K-1}(1-q)+o(F).\label{eqn:lK}
\end{align}
Next, consider a set $\mathcal{S}\subset\mathcal{K}$ such that $|\mathcal{S}|=K-1$, $U_{k,\mathcal{S}\backslash\{k\}}$ is composed of two parts, i.e.,  $V_{k,\mathcal{S}\backslash\{k\}}$, and some remaining bits in $U_{k,\mathcal{K}\backslash\{k\}}$. Note from \eqref{eqn:lK} that the size of the remaining bits in $U_{k,\mathcal{S}\backslash\{k\}}$ is of size $o(F)$. Thus, by \eqref{eqn:Vk}, we have
\begin{align}
|U_{k,\mathcal{S}\backslash\{k\}}|=Fq^{K-2}(1-q)^2+o(F), \forall~k\in\mathcal{S},\notag
\end{align}
which gives
\begin{align}
l_{\mathcal{S}}=\min_{k\in\mathcal{S}}\left\{|U_{k,\mathcal{S}\backslash\{k\}}|\right\}=Fq^{K-2}(1-q)^2+o(F).\notag
\end{align}
Generally, we can recursively show that, the size of $U_{k,\mathcal{S}\backslash\{k\}}$ for any $k\in\mathcal{S}\subset\mathcal{K}$ is
\begin{align}
|U_{k,\mathcal{S}\backslash\{k\}}|=Fq^{|\mathcal{S}|-1}(1-q)^{K-|\mathcal{S}|+1}+o(F)\notag
\end{align}
and hence $l_{\mathcal{S}}$ is given by
\begin{align}
l_{\mathcal{S}}=Fq^{|\mathcal{S}|-1}(1-q)^{K-|\mathcal{S}|+1}+o(F), ~\forall ~\mathcal{S}\subset\mathcal{K}.\notag
\end{align}

Thus, the delivered bits in total is
\begin{align}
R_{\bm{d}}F&=\sum_{\mathcal{S}\subset\mathcal{K}, \mathcal{S}\neq\emptyset }l_{S}\notag\\
&=\sum_{s=1}^K \sum_{\mathcal{S}\subset\mathcal{K},|\mathcal{S}|=s}l_{\mathcal{S}}\notag\\
&=\sum_{s=1}^K{K\choose s}q^{s-1}(1-q)^{K-s+1}+o(F)\notag\\
&=F\cdot\frac{1-q}{q}\left(1-(1-q)^K\right)+o(F).\notag
\end{align}
 Therefore, the average rate is given by
 \begin{align}
 \mathbb{E}[R_{\bm d}]=\frac{1-q}{q}\left(1-(1-q)^K\right)+\frac{o(F)}{F}.\notag
 \end{align}
Then, when $F\rightarrow\infty$, we have $$\mathbb{E}[R_{\bm d}]\rightarrow \frac{1-q}{q}\left(1-(1-q)^K\right),$$
which is consistent with \eqref{ERa1}, when $q_i=q=\frac{M}{N}$ for  $i=1,2,\cdots,K$.
\end{proof}

\end{appendices}

\bibliographystyle{IEEEtran}

\end{document}